\documentclass[conference]{IEEEtran}
\IEEEoverridecommandlockouts
\usepackage{amsmath,amssymb,amsfonts}
\usepackage{color}
\usepackage{xcolor}
\usepackage{bm}
\usepackage{graphicx}
\usepackage{multirow}
\usepackage{algorithm}
\usepackage{algorithmic}
\usepackage{booktabs}
\usepackage{array}
\usepackage{amsthm}
\usepackage{lipsum}
\usepackage{enumerate}
\usepackage{stfloats}
\usepackage{subfigure}
\usepackage{cases}
\usepackage{cite}
\usepackage{mathtools}
\usepackage{cuted}
\def\C{\mathbb{C}}

\def\Y{\mathbf{Y}}

\allowdisplaybreaks

\newtheorem{remark}{Remark}

\newtheorem{proposition}{Proposition}

\newtheoremstyle{noparens}%
  {}{}%
  {\itshape}{}%
  {\bfseries}{.}%
  { }%
  {\thmname{#1}\thmnumber{ #2}\mdseries\thmnote{ #3}}

\def\BibTeX{{\rm B\kern-.05em{\sc i\kern-.025em b}\kern-.08em
    T\kern-.1667em\lower.7ex\hbox{E}\kern-.125emX}}
\begin{document}

\title{Hybrid Digital and Microwave Linear Analog Computer (MiLAC)-aided Beamforming for Multiuser MIMO-OFDM Systems\\
}

\author{\IEEEauthorblockN{Yiyang Peng, Zheyu Wu, and Bruno Clerckx}
\IEEEauthorblockA{\textit{Department of Electrical and Electronic Engineering} \\
\textit{Imperial College London}\\
London SW7 2AZ, U.K. \\
Email: \{yiyang.peng22, zheyu.wu, b.clerckx\}@imperial.ac.uk}
}

\maketitle

\begin{abstract}
Microwave linear analog computing (MiLAC) has recently emerged as a promising architecture for analog-domain beamforming. In particular, a hybrid digital-MiLAC architecture was proposed and was shown to achieve fully-digital beamforming flexibility in narrowband systems when the number of RF chains equals the number of data streams. However, its performance in wideband systems remains unexplored. This paper presents the first study of hybrid digital-MiLAC beamforming for wideband multi-user multiple-input single-output (MU-MISO) systems. We first characterize the minimum number of radio-frequency (RF) chains required for hybrid digital-MiLAC beamforming to realize an arbitrary set of fully-digital beamforming matrices across all subcarriers. It turns out that, unlike in the narrowband case, a larger number of RF chains is generally required in frequency-selective channels to achieve fully-digital beamforming flexibility, which may be unfavorable in practice. To study the performance of hybrid digital-MiLAC beamforming with a limited number of RF chains, we then formulate the average sum-rate maximization problem and develop an efficient weighted minimum mean-square error (WMMSE)-based algorithm for beamforming design. Simulation results show that hybrid digital-MiLAC beamforming consistently outperforms conventional hybrid digital-analog beamforming, and achieves $89.93\%$ of the fully-digital sum-rate while using only $12.5\%$ of the RF chains in highly frequency-selective channels.
\end{abstract}


\section{Introduction}
To meet the requirements of future sixth-generation (6G) wireless systems, wireless networks are evolving toward upper mid-band frequencies and very large-scale antenna arrays, giving rise to massive or gigantic multiple-input multiple-output (MIMO) systems \cite{6G}. However, in conventional fully-digital beamforming architectures, each antenna is connected to one radio-frequency (RF) chain. When the number of antennas becomes very large, the hardware cost, power consumption, and signal processing complexity become prohibitive. In order to reduce the number of RF chains, hybrid digital-analog beamforming architecture has been widely studied \cite{hybrid_ofdm,hybrid_narrowband}. In hybrid beamforming, a low-dimensional digital beamformer is followed by an analog RF network, usually implemented by phase shifters. This can greatly reduce the RF-chain requirement, but the phase shifter network imposes structural constraints and therefore limits its design flexibility.

Recently, microwave linear analog computer (MiLAC) has emerged as a novel architecture that can implement linear transformations directly in the analog domain \cite{nerini_milac_partI,nerini_milac_partII}. A MiLAC is a reconfigurable multiport microwave network composed of tunable components, which enables a broad class of linear operations with significantly reduced computational complexity \cite{nerini_milac_partI,nerini_milac_partII,qiaosen_milac_ce}. When applied to beamforming, MiLAC processes RF signals through a multiport network whose input ports are connected to RF chains and output ports are connected to transmit antennas \cite{nerini_milac_capacity,nerini_milac_stem,zheyu_milac_beamforming,tianyu_milac_beamforming}. 
Under practical lossless and reciprocal constraints, MiLAC-aided beamforming was shown to achieve the same capacity as digital beamforming in point-to-point MIMO systems \cite{nerini_milac_capacity}. However, for multi-user multiple-input single-output (MU-MISO) systems, it generally cannot realize arbitrary digital beamforming matrices \cite{zheyu_milac_beamforming,tianyu_milac_beamforming}. To overcome this limitation, a hybrid digital-MiLAC architecture was proposed in \cite{zheyu_milac_beamforming}, and was shown to achieve fully-digital beamforming flexibility when the number of RF chains equals the number of data streams.

However, existing MiLAC works are limited to narrowband systems \cite{nerini_milac_partI,nerini_milac_partII,nerini_milac_capacity,nerini_milac_stem,qiaosen_milac_ce,zheyu_milac_beamforming,tianyu_milac_beamforming}, and extending MiLAC-aided beamforming to wideband systems is nontrivial. In a wideband system, the digital beamforming matrices may vary significantly across subcarriers. In contrast, under the MiLAC model considered in existing works \cite{nerini_milac_capacity,nerini_milac_stem,zheyu_milac_beamforming,tianyu_milac_beamforming}, the same MiLAC beamforming matrix is shared by all subcarriers. For frequency-selective channels, such a common MiLAC beamforming matrix is generally insufficient to realize the flexibility of fully-digital beamforming. Therefore, the conclusions established for narrowband systems do not directly apply to wideband systems. 

Motivated by the above limitation, this paper investigates the performance of hybrid digital-MiLAC architecture in wideband systems, where a low-dimensional digital beamformer adapts the beamforming across subcarriers, followed by a lossless, reciprocal, and frequency-flat  MiLAC that provides the main analog beamforming transformation. The main contributions are summarized as follows. \textit{First}, we characterize the minimum number of RF chains required for hybrid digital-MiLAC beamforming to realize an arbitrary set of fully-digital beamforming matrices over all subcarriers, revealing a fundamental difference from the narrowband system. \textit{Second}, we formulate the average sum-rate maximization problem for the considered wideband orthogonal frequency division multiplexing (OFDM) system and develop an efficient weighted minimum mean-square error (WMMSE)-based algorithm for beamforming design. \textit{Third}, we present simulation results to demonstrate the effectiveness and performance tradeoff of the hybrid digital-MiLAC beamforming compared to other schemes in wideband systems.

\section{System Model}

\subsection{System Model}

\begin{figure*}
    \centering
    \includegraphics[width=0.8\textwidth]{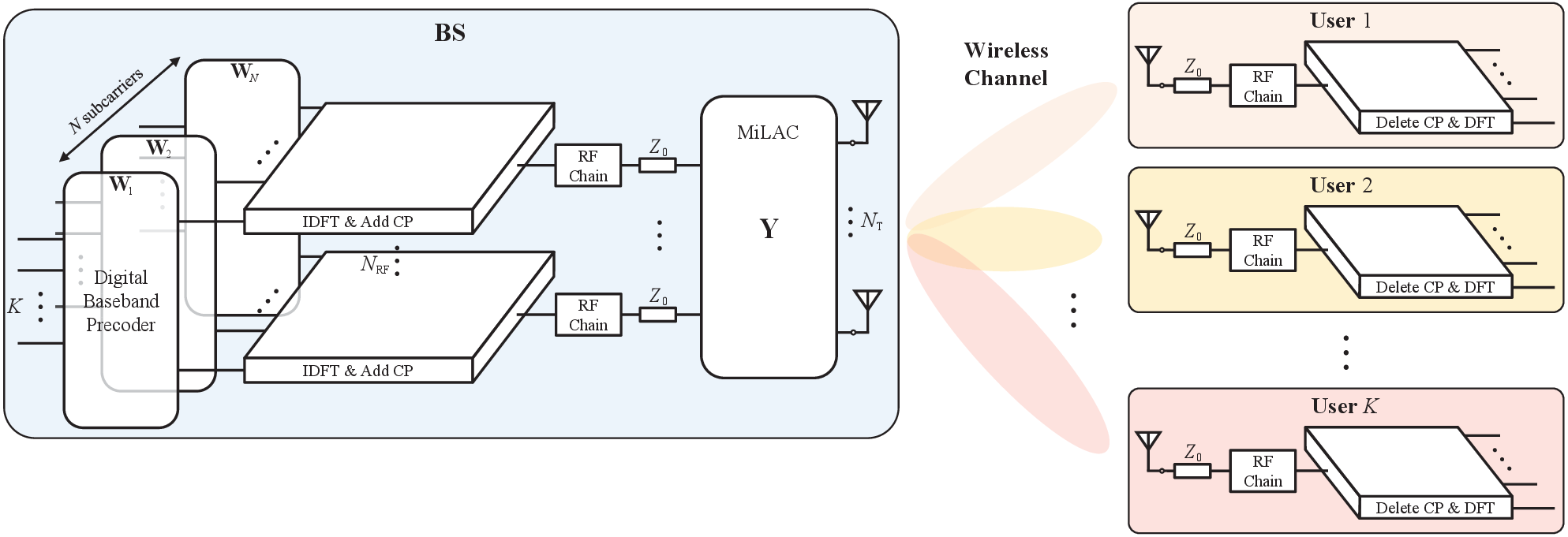}
    \caption{The illustration of a hybrid digital-MiLAC-aided MU-MISO-OFDM system.}\label{fig:system_model}
\end{figure*}

We consider a wideband multi-user multiple-input single-output (MU-MISO) system, employing an orthogonal frequency division multiplexing (OFDM) transmission scheme with $N$ subcarriers and a hybrid digital-MiLAC beamforming architecture, as shown in Fig. \ref{fig:system_model}. The base station (BS) is equipped with $N_\text{T}$ transmit antennas and $N_\text{RF}$ RF chains, and serves $K$ single-antenna users which are scheduled over the entire bandwidth. Let $\mathbf{s}_n$ denote the data symbol vector for all users on the $n$-th subcarrier, where $\mathbf{s}_n= [s_{1,n},s_{2,n},\ldots,s_{K,n}]^T\in\C^{K\times 1},~\forall n\in\mathcal{N}=\{1,2,\ldots,N\}$. Here, $s_{k,n}$ is the information symbol intended for user $k$ on subcarrier $n$, and we assume $\mathbb{E}\{\mathbf{s}_n\mathbf{s}_n^H\}=\mathbf{I}_K,~\forall n\in\mathcal{N}$. At each subcarrier $n$, the BS first precodes $\mathbf{s}_n$ by a digital precoder $\mathbf{W}_n\in\C^{N_\text{RF}\times K}$ in the frequency domain, given by $\mathbf{W}_n = [\mathbf{w}_{1,n},\mathbf{w}_{2,n},\ldots,\mathbf{w}_{K,n}]$, where $\mathbf{w}_{k,n}\in\C^{N_\text{RF}\times 1}$ denotes the digital precoder for user $k$ on the $n$-th subcarrier, $\forall k\in\mathcal{K}=\{1,2,\ldots,K\},~\forall n\in\mathcal{N}$. The overall digitally precoded frequency-domain signal $\mathbf{s}'\in\C^{N_\text{RF}N\times 1}$ can be written as
\begin{equation}
    \mathbf{s}'=\mathbf{W}\mathbf{s},
\end{equation}
where $\mathbf{W}=\text{blkdiag}(\mathbf{W}_1,\mathbf{W}_2,\ldots,\mathbf{W}_N)$, and $\mathbf{s}=[\mathbf{s}_1^T,\mathbf{s}_2^T,\ldots,\mathbf{s}_N^T]^T$. Then $\mathbf{s}'$ is converted to the time domain by performing $N_\text{RF}$ $N$-point inverse discrete Fourier transforms (IDFTs), which yields the overall time-domain signal $\bar{\mathbf{s}}\in\C^{N_\text{RF}N\times 1}$, as
\begin{equation}
    \bar{\mathbf{s}}=(\mathbf{F}^H \otimes \mathbf{I}_{N_\text{RF}})\mathbf{s}',\label{eq:s_bar}
\end{equation}
where $\mathbf{F}\in\C^{N\times N}$ is the normalized DFT matrix defined as $[\mathbf{F}]_{i,j}\triangleq\frac{1}{\sqrt{N}}e^{\frac{-\jmath 2\pi}{N}(i-1)(j-1)},~\forall i,j\in\mathcal{N}$. After insertion of the cyclic prefix (CP)\footnote{We assume perfect CP insertion and removal at the BS and users, respectively, with a CP length no smaller than the maximum delay of the discrete-time
channel impulse response. Hence, for notational simplicity, we only model the useful OFDM block and omit the explicit CP-appended signal.} to $\bar{\mathbf{s}}$, the digital baseband signal is converted to analog and up-converted to passband through the radio-frequency (RF) chains. As shown in Fig. \ref{fig:system_model}, each MiLAC input port is driven by an RF chain and is referenced to the impedance $Z_0=50\ \Omega$. In this paper, we assume that the response of the MiLAC is approximately constant over the occupied OFDM bandwidth. Accordingly, the MiLAC can be modeled by a matrix $\mathbf{P}^\text{M}\in\C^{N_\text{T}\times N_\text{RF}}$, which characterizes the equivalent baseband input-output relationship\footnote{Although the MiLAC physically processes passband signals, throughout this paper we adopt the standard equivalent complex-valued baseband representation of a linear microwave network.} of the MiLAC \cite{nerini_milac_partI,nerini_milac_partII}. Under this frequency-flat assumption, at each time instant the MiLAC applies the same matrix $\mathbf{P}^{\text{M}}$ to the $N_{\text{RF}}$-dimensional time-domain sample vector of the OFDM block $\bar{\mathbf{s}}$, i.e.,
\begin{equation}
    \bar{\mathbf{x}}=(\mathbf{I}_N \otimes \mathbf{P}^\text{M})\bar{\mathbf{s}},\label{eq:x_time}
\end{equation}
where $\bar{\mathbf{x}}=[\bar{\mathbf{x}}_1^T,\bar{\mathbf{x}}_2^T,\ldots,\bar{\mathbf{x}}_N^T]^T\in\C^{N_\text{T}N\times 1}$ denotes the MiLAC output of an OFDM block, which is transmitted by the BS, with $\bar{\mathbf{x}}_m\in\C^{N_\text{T}\times 1},~\forall m\in\mathcal{N}$, denoting the $m$-th time sample of the OFDM signal. Equivalently, in the frequency domain, the same matrix $\mathbf{P}^{\text{M}}$ is shared by all subcarriers. The corresponding frequency-domain transmit vector is denoted by $\mathbf{x}\in\C^{N_\text{T}N\times 1}$, which is given by
\begin{align}
    \mathbf{x}&=(\mathbf{F} \otimes \mathbf{I}_{N_\text{T}})\bar{\mathbf{x}}\notag\\
    &=(\mathbf{F} \otimes \mathbf{I}_{N_\text{T}})(\mathbf{I}_N \otimes \mathbf{P}^\text{M})(\mathbf{F}^H \otimes \mathbf{I}_{N_\text{RF}})\mathbf{W}\mathbf{s}\notag\\
    &=(\mathbf{I}_N \otimes \mathbf{P}^\text{M})\mathbf{W}\mathbf{s}.\label{eq:x_freq}
\end{align}
Hence, the frequency-domain transmitted signal on the $n$-th subcarrier is
\begin{equation}
\mathbf{x}_n=\mathbf{P}^\text{M}\mathbf{W}_n\mathbf{s}_n,~\forall n\in\mathcal{N},
\end{equation}
where $\mathbf{x}_n=[\mathbf{x}]_{(n-1)N_\text{T}+1:n N_\text{T}}$. Assuming perfect CP insertion and removal, the multipath channel convolution over the useful OFDM block becomes circular and can therefore be diagonalized by the DFT, leading to parallel subcarrier channels \cite{mimo_ofdm}. For user $k$, the received signal on subcarrier $n$ is given by
\begin{equation}
    r_{k,n}=\mathbf{h}_{k,n}^H\mathbf{P}^\text{M} \mathbf{w}_{k,n}s_{k,n}+ \sum_{j\neq k}\mathbf{h}_{k,n}^H\mathbf{P}^\text{M} \mathbf{w}_{j,n}s_{j,n}+z_{k,n},
\end{equation}
where $\mathbf{h}_{k,n}\in\C^{N_\text{T}\times 1}$ denotes the frequency-domain channel from the BS to user $k$ on the $n$-th subcarrier, and $z_{k,n}\sim\mathcal{CN}(0,\sigma_{k,n}^2)$ is the additive white Gaussian noise (AWGN). 

\subsection{MiLAC Modeling}
A MiLAC can be modeled as a reconfigurable microwave network \cite{nerini_milac_partI,nerini_milac_partII}. In the considered hybrid digital-MiLAC beamforming architecture, the MiLAC is represented as an $(N_{\text{RF}}+N_{\text{T}})$-port network, where the first $N_{\text{RF}}$ ports are connected to the RF chains and the remaining $N_{\text{T}}$ ports are connected to the transmit antennas. In this paper, we consider the fully-connected MiLAC architecture studied in \cite{nerini_milac_partI,nerini_milac_partII,nerini_milac_capacity,zheyu_milac_beamforming}, where each port is connected to the ground through a tunable component and is also interconnected to every other port through tunable components. Define $\Y\in\C^{(N_{\text{RF}}+N_{\text{T}})\times(N_{\text{RF}}+N_{\text{T}})}$ as the admittance matrix of MiLAC, according to \cite{nerini_milac_partI,nerini_milac_partII}, $\mathbf{P}^\text{M}$ is determined by $\mathbf{Y}$ as
\begin{equation}
    \mathbf{P}^\text{M} = \left[(\mathbf{I}_{N_{\text{RF}}+N_{\text{T}}}+Z_0\mathbf{Y})^{-1}\right]_{N_{\text{RF}}+1:N_{\text{RF}}+N_{\text{T}},1:N_{\text{RF}}}.
\end{equation}
Physically, the MiLAC can be implemented using tunable admittance components, which gives rise to a reconfigurable admittance matrix $\mathbf{Y}$. Equivalently, the MiLAC can also be characterized by its scattering matrix $\bm\Phi\in\C^{(N_{\text{RF}}+N_{\text{T}})\times(N_{\text{RF}}+N_{\text{T}})}$, which is related to $\mathbf{Y}$ through
\begin{equation}
    \bm\Phi=(\mathbf{I}_{N_{\text{RF}}+N_{\text{T}}}+Z_0\mathbf{Y})^{-1}(\mathbf{I}_{N_{\text{RF}}+N_{\text{T}}}-Z_0\mathbf{Y}).
\end{equation}
As shown in \cite{nerini_milac_capacity,nerini_milac_stem}, the relationship between $\mathbf{P}^\text{M}$ and $\bm\Phi$ can be further simplified as
\begin{equation}
    \mathbf{P}^\text{M} = \frac{1}{2}[\bm\Phi]_{N_{\text{RF}}+1:N_{\text{RF}}+N_{\text{T}},1:N_{\text{RF}}}.
\end{equation}
Moreover, the tunable components employed in the MiLAC are typically reciprocal devices, e.g., varactors. In this paper, we consider the reciprocal and lossless MiLAC \cite{nerini_milac_capacity,nerini_milac_stem,zheyu_milac_beamforming}, which mathematically results in a symmetric and unitary scattering matrix \cite{pozar2009microwave}, respectively, i.e.,
\begin{equation}
\bm\Phi=\bm\Phi^T,\bm\Phi^H\bm\Phi=\mathbf{I}_{N_{\text{RF}}+N_{\text{T}}}.
\end{equation}

\section{Hybrid Digital-MiLAC-Aided Beamforming Design For The MU-MISO-OFDM System}
\subsection{Problem Formulation}
Define $\mathbf{P}\triangleq 2\mathbf{P}^\text{M}$ for notation simplicity. A key result established in \cite{zheyu_milac_beamforming} is that the lossless and reciprocal constraints of MiLAC on the scattering matrix $\bm{\Phi}$ can be equivalently expressed as a spectral-norm constraint on the matrix $\mathbf{P}$, i.e.,
\begin{equation}
    \|\mathbf{P}\|_2 \leq 1.
\end{equation}
The achievable rate for user $k$ on the $n$-th subcarrier can be expressed as
\begin{equation}
    R_{k,n}= \log_2 \left(1+\frac{\frac{1}{4}|\mathbf{h}_{k,n}^H\mathbf{P} \mathbf{w}_{k,n}|^2}{\frac{1}{4}\sum_{j\neq k}|\mathbf{h}_{k,n}^H\mathbf{P} \mathbf{w}_{j,n}|^2+\sigma_{k,n}^2}\right),
\end{equation}
where the factor $\frac{1}{4}$ is caused by $\mathbf{P}= 2\mathbf{P}^\text{M}$. We aim to maximize the average sum-rate for the MU-MISO-OFDM system with hybrid digital-MiLAC beamforming architecture subject to the constraints of lossless and reciprocal MiLAC and digital power budget. The problem is formulated as
\begin{subequations}\label{eq:max_sum_rate_primal}
    \begin{align}                    \max_{\mathbf{P},\{\mathbf{W}_n\}_{n=1}^N}~& \frac{1}{N}\sum_{n=1}^N\sum_{k=1}^K  \log_2 \left(1+\frac{|\mathbf{h}_{k,n}^H\mathbf{P} \mathbf{w}_{k,n}|^2}{\sum\limits_{j\neq k}|\mathbf{h}_{k,n}^H\mathbf{P} \mathbf{w}_{j,n}|^2+\tilde{\sigma}_{k,n}^2}\right)\label{eq:max_sum_rate_primal_a}\\
    \mathrm{s.t.}~~~~~&\|\mathbf{P}\|_2\leq 1,\label{eq:max_sum_rate_primal_b}\\
    &\sum_{n=1}^N \|\mathbf{W}_n\|_F^2 \leq P_\text{T},\label{eq:max_sum_rate_primal_c}
    \end{align}
\end{subequations}
where $\tilde{\sigma}_{k,n}^2\triangleq 4\sigma_{k,n}^2,~\forall n\in\mathcal{N},~\forall k\in\mathcal{K}$, and $P_\text{T}$ denotes the total digital power at the BS.

\subsection{Minimum Number of RF Chains to Realize Fully-Digital Beamformers}

In the narrowband multi-user MISO system, it was established in \cite{zheyu_milac_beamforming} that MiLAC-aided beamforming cannot realize any fully-digital beamforming matrix, while hybrid digital-MiLAC beamforming can realize any fully-digital beamforming matrix with the same number of RF chains as the number of streams. Extending this characterization to wideband OFDM systems is nontrivial, because the  digital beamforming matrix can be designed independently for each subcarrier, whereas the MiLAC response $\mathbf{P}^\text{M}$ is shared across all subcarriers.
 This naturally leads to the following question: what is the minimum number of RF chains required by hybrid digital-MiLAC beamforming to realize an arbitrary set of fully-digital beamforming matrices in the wideband case? We provide the answer in the following Proposition~\ref{proposition:hybrid_rf_chain_digital}.

\begin{proposition}\label{proposition:hybrid_rf_chain_digital}
Hybrid digital-MiLAC beamforming can realize any fully digital beamforming matrices $\{\mathbf{W}^\textup{D}_n\}_{n=1}^N$ if and only if the number of RF chains in the hybrid digital-MiLAC architecture satisfies $N_\text{RF}\geq\min\{N_\text{T},KN\}$.
\end{proposition}
\begin{proof}
We introduce the following notations for the hybrid digital-MiLAC beamforming matrix. Let $\mathbf{B}_n=\mathbf{P}\mathbf{W}_n\in\C^{N_\text{T}\times K}$, $\tilde{\mathbf{B}}=[\mathbf{B}_1,\mathbf{B}_2,\ldots,\mathbf{B}_N]\in\C^{N_\text{T}\times KN}$, and $\tilde{\mathbf{W}}=[\mathbf{W}_1,\mathbf{W}_2,\ldots,\mathbf{W}_N]\in\C^{N_\text{RF}\times KN}$, then we have $\tilde{\mathbf{B}}=\mathbf{P}\tilde{\mathbf{W}}$ and $\sum_{n=1}^N \|\mathbf{W}_n\|_F^2= \|\tilde{\mathbf{W}}\|_F^2$. 

We first prove the sufficiency. Suppose that $N_\text{RF}\geq\min\{N_\text{T},KN\}$. For any digital beamforming matrix $\mathbf{W}^\text{D}=[\mathbf{W}^\text{D}_1,\mathbf{W}^\text{D}_2,\ldots,\mathbf{W}^\text{D}_N]\in\C^{N_\text{T}\times KN}$ with $\|\mathbf{W}^\text{D}\|_F^2\leq P_\text{T}$, let $\mathbf{W}^\text{D}=\mathbf{U}\bm\Sigma\mathbf{V}^H$ be its compact SVD, where $\mathbf{U}\in\C^{N_\text{T}\times l}$, $\bm\Sigma\in\C^{l\times l},\mathbf{V}\in\C^{KN\times l}$, with $l=\min\{N_\text{T},KN\}$. We show that $\mathbf{W}^{\text{D}}$ can be achieved by the hybrid digital-MiLAC architecture by letting  $\mathbf{P}=\mathbf{U}$ and $\tilde{\mathbf{W}}=\bm\Sigma\mathbf{V}^H$. This specification immediately yields $\mathbf{W}^{\text{D}}=\mathbf{P}\tilde{\mathbf{W}}$. 
Note that  $\mathbf{U}^H\mathbf{U}=\mathbf{I}_l$ and $\|\bm\Sigma\mathbf{V}^H\|_F^2=\|\mathbf{U}^H\mathbf{W}^\text{D}\|_F^2=\|\mathbf{U}\mathbf{U}^H\mathbf{W}^\text{D}\|_F^2=\|\mathbf{W}^\text{D}\|_F^2$. It follows that $\|\mathbf{P}\|_2=1$, which satisfies the MiLAC constraint in \eqref{eq:max_sum_rate_primal_b},  and $\|\tilde{\mathbf{W}}\|_F^2\leq P_\text{T}$, which satisfies the digital-power constraint in \eqref{eq:max_sum_rate_primal_c}. Hence, we show that any $\mathbf{W}^\text{D}$ can be realized when $N_\text{RF}=\min\{N_\text{T},KN\}$. 

Next, we prove the necessity. Suppose $N_\text{RF}<\min\{N_\text{T},KN\}$, then $\text{rank}(\mathbf{P}\tilde{\mathbf{W}})\leq N_\text{RF}<l$. Therefore, any full-rank digital beamforming matrix $\mathbf{W}^\text{D}$ cannot be realized by $\mathbf{P}\tilde{\mathbf{W}}$. That completes the proof.
\end{proof}
Proposition 1 generalizes the result in \cite{zheyu_milac_beamforming} to the wideband case. When $N=1$, it reduces to the result in \cite{zheyu_milac_beamforming}. However, in general, $KN>N_\text{T}$, thus realizing arbitrary fully-digital beamforming matrices in OFDM systems requires $N_\text{RF}\geq N_\text{T}$. In this sense, hybrid digital-MiLAC beamforming does not offer RF-chain reduction in general wideband systems. 
In the following remark, we show that the required number of RF chains can be further reduced to realize   rank-deficient beamforming matrices.
\begin{remark}\label{remark1}
To achieve a specific digital beamforming matrix $\mathbf{W}^{\textup{D}}$, the required number of RF chains is $\text{rank}(\mathbf{W}^\textup{D})$, which can be smaller than the number given in Proposition~\ref{proposition:hybrid_rf_chain_digital} when $\mathbf{W}^{\textup{D}}$ is rank deficient. 
 Specifically, for a given matrix $\mathbf{W}^\textup{D},$ it can always be decomposed as $\mathbf{W}^\textup{D}=\mathbf{U}_W\bm\Sigma_W\mathbf{V}_W^H$, where $\mathbf{U}_W\in\C^{N_\text{T}\times l'}$, $\bm\Sigma_W\in\C^{l'\times l'},\mathbf{V}_W\in\C^{KN\times l'}$, with $l'=\text{rank}(\mathbf{W}^\textup{D})\leq \min\{N_\text{T},KN\}$. Let $N_\text{RF}=l'$, by setting $\mathbf{P}=\mathbf{U}_W$ and $\tilde{\mathbf{W}}=\bm\Sigma_W\mathbf{V}_W^H$, we can use the same proof in Proposition 1 to show that $\mathbf{P}$ satisfies \eqref{eq:max_sum_rate_primal_b} and $\tilde{\mathbf{W}}$ satisfies \eqref{eq:max_sum_rate_primal_c}. Therefore, $\mathbf{W}^\textup{D}$ can be realized by the hybrid digital-MiLAC using $l'$ RF chains.
\end{remark}
An important implication of Remark \ref{remark1} is that, as the OFDM system becomes more frequency-flat, fewer RF chains are required to achieve fully-digital performance. This is because optimal beamforming matrices $\{\mathbf{W}_n^\text{D}\}$ across subcarriers become more similar for frequency-flat channels, and thus the aggregated matrix $\mathbf{W}^{\text{D}}$ is more likely to be rank-deficient. 

To evaluate the performance of the hybrid digital-MiLAC beamforming scheme with a general number of RF chains, we propose an efficient algorithm for solving \eqref{eq:max_sum_rate_primal} in the following subsection.

\subsection{WMMSE Approach of Solving \eqref{eq:max_sum_rate_primal}}
We apply the weighted minimum mean-square error (WMMSE) algorithm to solve the average sum-rate maximization problem in \eqref{eq:max_sum_rate_primal}. Denote $u_{k,n}\in\C$ as the receive equalizer for user $k$ on subcarrier $n$. Define $\tilde{r}_{k,n}\triangleq 2r_{k,n}$, which is given by
\begin{equation}
\tilde{r}_{k,n}=\mathbf{h}_{k,n}^H\mathbf{P}\mathbf{W}_n\mathbf{s}_n+\tilde{z}_{k,n},~\forall n\in\mathcal{N},~\forall k\in\mathcal{K},
\end{equation}
where $\tilde{z}_{k,n}=2z_{k,n}$.
The estimated symbol of user $k$ on subcarrier $n$ is given by $\hat{s}_{k,n}=u_{k,n}^*\tilde{r}_{k,n}$. Then, the mean-square error (MSE) between symbol $s_{k,n}$ and estimated symbol $\hat{s}_{k,n}$, i.e., $E_{k,n}\triangleq \mathbb{E}\{|\hat{s}_{k,n}-s_{k,n}|^2\}$, can be expressed as 
\begin{align}
    E_{k,n}(u_{k,n},\mathbf{P},\mathbf{W}_{n})
    &=|1-u_{k,n}^*\mathbf{h}_{k,n}^H\mathbf{P}\mathbf{w}_{k,n}|^2 \notag\\&+ |u_{k,n}|^2\left(\sum_{j\neq k}|\mathbf{h}_{k,n}^H\mathbf{P}\mathbf{w}_{j,n}|^2+\tilde{\sigma}^2\right).
\end{align}
By introducing weight auxiliary variables $\{\bm\omega_n\}_{n=1}^N$, and applying the transformation in \cite{WMMSE}, we can equivalently reformulate problem \eqref{eq:max_sum_rate_primal} into the following WMMSE form:
\begin{subequations}\label{eq:max_sum_rate_wmmse}
    \begin{align}                    \min_{\mathbf{P},\{\mathbf{W}_n\},\{\mathbf{u}_n\},\{\bm\omega_n\}}~& \frac{1}{N}\sum_{n=1}^N\sum_{k=1}^K  \left(\frac{1}{\ln 2}\omega_{k,n}E_{k,n}-\log_2(\omega_{k,n})\right)\label{eq:max_sum_rate_wmmse_a}\\
    \mathrm{s.t.}~~~~~~~~~&\text{\eqref{eq:max_sum_rate_primal_b}, \eqref{eq:max_sum_rate_primal_c}},\label{eq:max_sum_rate_wmmse_b}
    \end{align}
\end{subequations}
where $\mathbf{u}_n=[u_{1,n},u_{2,n},\ldots,u_{K,n}]\in\C^{K\times 1}$ and $\bm\omega_n=[\omega_{1,n},\omega_{2,n},\ldots,\omega_{K,n}]\in\mathbb{R}^{K\times 1}$ with $\omega_{k,n}>0,~\forall n\in\mathcal{N}$. Since problem \eqref{eq:max_sum_rate_wmmse} is convex with respect to each variable block, we apply the block coordinate descent (BCD) method to minimize \eqref{eq:max_sum_rate_wmmse_a}. Next, we discuss the subproblem of each variable block separately. 

\subsubsection{Update of $\{\mathbf{u}_n\}$ and $\{\bm\omega_n\}$}
By setting the partial derivatives of the objective function with respect to $u_{k,n}$ and $\omega_{k,n}$ to zero, respectively, while fixing the other variables, we obtain the updates of $\mathbf{u}_n$ and $\bm\omega_n$:
\begin{equation}
    u_{k,n}=\frac{\mathbf{h}_{k,n}^H\mathbf{P}\mathbf{w}_{k,n}}{\sum_{j=1}^K|\mathbf{h}_{k,n}^H\mathbf{P}\mathbf{w}_{j,n}|^2+\tilde{\sigma}^2},\label{eq:u-subproblm_solution}
\end{equation}
and
\begin{equation}
    \omega_{k,n}=(1-u_{k,n}^*\mathbf{h}_{k,n}^H\mathbf{P}\mathbf{w}_{k,n})^{-1},~\forall n\in\mathcal{N},~\forall k\in\mathcal{K}.\label{eq:omega-subproblm_solution}
\end{equation}
\subsubsection{Update of $\{\mathbf{W}_n\}$}
For fixed $\{\mathbf{u}_n\}$, $\{\bm\omega_n\}$ and $\mathbf{P}$, digital precoders $\{\mathbf{W}_n\}_{n=1}^N$ are updated by solving the following subproblem:
\begin{subequations}\label{eq:max_sum_rate_wmmse_W}
    \begin{align}                    \min_{\{\mathbf{W}_n\}}~& \sum_{n=1}^N\sum_{k=1}^K  \left(\mathbf{w}_{k,n}^H\mathbf{Q}_n\mathbf{w}_{k,n}-2\Re\{\mathbf{a}_{k,n}^H\mathbf{w}_{k,n}\}\right)\label{eq:max_sum_rate_wmmse_W_a}\\\mathrm{s.t.}~~&\text{\eqref{eq:max_sum_rate_primal_c}},\label{eq:max_sum_rate_wmmse_W_b}
    \end{align}
\end{subequations}
where $\mathbf{Q}_n\triangleq\mathbf{P}^H\left( \sum_{j=1}^K \omega_{j,n} |u_{j,n}|^2 \mathbf{h}_{j,n} \mathbf{h}_{j,n}^H\right)\mathbf{P}$ and $\mathbf{a}_{k,n}\triangleq\omega_{k,n}u_{k,n}\mathbf{P}^H\mathbf{h}_{k,n},~\forall n\in\mathcal{N},~\forall k\in\mathcal{K}$. Problem \eqref{eq:max_sum_rate_wmmse_W} is a convex quadratically constrained quadratic program and can be optimally solved using the Karush-Kuhn-Tucker (KKT) conditions \cite{WMMSE}. 

\subsubsection{Update of $\mathbf{P}$}
For fixed $\{\mathbf{u}_n\}$, $\{\bm\omega_n\}$, and $\{\mathbf{W}_n\}$, the MiLAC matrix $\mathbf{P}$ is updated by solving the following subproblem:
\begin{subequations}\label{eq:max_sum_rate_wmmse_P}
    \begin{align}                    \min_{\mathbf{P}}~&\sum_{n=1}^N \left( \text{Tr}(\mathbf{P}^H \mathbf{M}_n \mathbf{P} \mathbf{C}_n) - 2\Re\{ \text{Tr}(\mathbf{D}_n\mathbf{P}) \} \right)\label{eq:max_sum_rate_wmmse_P_a}\\\mathrm{s.t.}~&\text{\eqref{eq:max_sum_rate_primal_b}},\label{eq:max_sum_rate_wmmse_P_b}
    \end{align}
\end{subequations}
where $\mathbf{M}_n \triangleq \sum_{k=1}^K \omega_{k,n} |u_{k,n}|^2 \mathbf{h}_{k,n} \mathbf{h}_{k,n}^H$, $\mathbf{C}_n \triangleq \mathbf{W}_n \mathbf{W}_n^H$ and $\mathbf{D}_n \triangleq \sum_{k=1}^K \omega_{k,n} u_{k,n}^* \mathbf{w}_{k,n}\mathbf{h}_{k,n}^H$. Problem \eqref{eq:max_sum_rate_wmmse_P} has a quadratically objective function with a convex constraint. To efficiently solve it, we adopt the projected gradient descent (PGD) method. Let $\mathcal{P}\triangleq \{ \mathbf{P}\in\C^{N_\text{T}\times N_\text{RF}} \mid \|\mathbf{P}\|_2 \le 1 \}$ denote the feasible set defined by the spectral-norm constraint. The PGD algorithm updates $\mathbf{P}$ as
\begin{equation}
    \mathbf{P}^{(t+1)} = \text{Proj}_{\mathcal{P}} \left( \mathbf{P}^{(t)} - \eta \text{Grad}(\mathbf{P}^{(t)}) \right),
\end{equation}
where $\mathbf{P}^{(t)}$ denotes the value of $\mathbf{P}$ at the $t$-th iteration, $\text{Grad}(\mathbf{P}^{(t)})$ is the gradient of the objective function in \eqref{eq:max_sum_rate_wmmse_P_a} with respect to $\mathbf{P}$ at point $\mathbf{P}^{(t)}$, $\text{Proj}_{\mathcal{P}}$ represents the projection onto set $\mathcal{P}$, and $\eta>0$ is the step size. By simple calculation, 
\begin{equation}
    \text{Grad}(\mathbf{P}^{(t)})=2\sum_{n=1}^N \left( \mathbf{M}_n \mathbf{P}^{(t)} \mathbf{C}_n - \mathbf{D}_n^H \right).
\end{equation}
For a given matrix $\mathbf{X}\in\C^{N_\text{T}\times N_\text{RF}}$, the projection $\text{Proj}_{\mathcal{P}}(\mathbf{X})$ is performed via singular value thresholding, which clips each singular value of $\mathbf{X}$ at $1$. Specifically, let the SVD of $\mathbf{X}$ be $\mathbf{X}=\mathbf{U}_X \mathbf{\Sigma}_X \mathbf{V}_X^H$,  where $\mathbf{\Sigma}_X = [\text{diag}(\sigma_1, \sigma_2, \ldots,\sigma_{N_\text{RF}}),\mathbf{0}_{N_\text{RF}\times (N_\text{T}-N_\text{RF})}]^T\in\C^{N_\text{T}\times N_\text{RF}}$ contains all singular values (assume that $N_\text{RF}\leq N_\text{T}$). The projection is  given by
\begin{equation}
    \text{Proj}_{\mathcal{P}}(\mathbf{X}) = \mathbf{U}_X \tilde{\mathbf{\Sigma}}_X \mathbf{V}_X^H,    
\end{equation}
where $\tilde{\mathbf{\Sigma}}_X = [\text{diag}(\tilde{\sigma}_1, \tilde{\sigma}_2, \ldots,\tilde{\sigma}_{N_\text{RF}}),\mathbf{0}_{N_\text{RF}\times (N_\text{T}-N_\text{RF})}]^T$ with $\tilde{\sigma}_m = \min\{1, \sigma_m\}$, $m=1,2,\ldots,N_\text{RF}$. 

For the convex $\mathbf{P}$-update subproblem \eqref{eq:max_sum_rate_wmmse_P}, the convergence of the PGD algorithm is guaranteed when the step size $\eta\leq L_G^{-1}$, where $L_G$ is the Lipschitz constant of $\text{Grad}(\mathbf{P})$. For any two MiLAC matrices $\mathbf{P}_1$ and $\mathbf{P}_2$,
\begin{align}
  &  \|\text{Grad}(\mathbf{P}_1)-\text{Grad}(\mathbf{P}_2)\|_F\notag\\
    &= \left\|2\sum_{n=1}^N\mathbf{M}_n (\mathbf{P}_1-\mathbf{P}_2) \mathbf{C}_n \right\|_F\notag\\
    &\leq2\sum_{n=1}^N\left\|\mathbf{M}_n (\mathbf{P}_1-\mathbf{P}_2) \mathbf{C}_n \right\|_F\notag\\
    &\leq2\sum_{n=1}^N\|\mathbf{M}_n \|_2 \|\mathbf{C}_n\|_2\|\mathbf{P}_1-\mathbf{P}_2  \|_F.
\end{align}
Therefore,  $L_G=2\sum_{n=1}^N\|\mathbf{M}_n \|_2\|\mathbf{W}_n\|_2^2$ is a Lipschitz constant of $\text{Grad}(\mathbf{P})$, and thus we set the stepsize as $\eta=1/(2\sum_{n=1}^N\|\mathbf{M}_n \|_2 \|\mathbf{W}_n\|_2^2)$.

At each iteration of BCD, every variable block is solved to global optimality. In particular, the 
$\{\mathbf{u}_n\}$- and $\{\bm{\omega}_n\}$-subproblems admit closed-form solutions in \eqref{eq:u-subproblm_solution} and \eqref{eq:omega-subproblm_solution}, respectively. The $\{\mathbf{W}_n\}$-subproblem is convex and has a unique solution \cite{WMMSE}. The $\mathbf{P}$-subproblem is also convex, and the PGD algorithm is guaranteed to converge to its global optimal solution under the specified stepsize $\eta$. Therefore, the WMMSE algorithm yields a monotonically non-increasing objective function  and is guaranteed to converge to a stationary point of problem \eqref{eq:max_sum_rate_primal} according to  \cite[Theorem 3]{WMMSE}.

\section{Simulation Results}
\subsection{Simulation Setup}
In the considered MU-MISO-OFDM system, we set the bandwidth $B$ to 300 MHz, the number of subcarriers, transmit antennas, and users as $N=64$, $N_\text{T}=64$, and $K=4$, respectively. The OFDM channels are modeled as a frequency-selective $D$-tap channels with $D=16$. The cyclic prefix length is set to 15. The time-domain channels of user $k$ is denoted as $\{\bar{\mathbf{h}}_{k,0},\bar{\mathbf{h}}_{k,1},\ldots,\bar{\mathbf{h}}_{k,D-1}\}$, where the $d$-th tap is given by $\bar{\mathbf{h}}_{k,d}\sim\mathcal{CN}(\mathbf{0},p_d\mathbf{I}_{N_\text{T}}),~d=0,1,\ldots,D-1$, and $\{p_d\}_{d=0}^{D-1}$ denotes the normalized power delay profile (PDP) that satisfies $\sum_{d=0}^{D-1}p_d=1$ and each $p_d \propto e^{-d/\epsilon}$, with $\epsilon$ being the PDP decay factor. The frequency-domain channel $\mathbf{h}_{k,n}\in\C^{N_\text{T}\times 1}$ is obtained by $\mathbf{h}_{k,n}=\sum_{d=0}^{D-1}\bar{\mathbf{h}}_{k,d}e^{-\jmath \frac{2\pi}{N}(n-1)d},~\forall n\in\mathcal{N},~\forall k\in\mathcal{K}$. To characterize the degree of frequency selectivity, we use the normalized root mean square (RMS) delay spread, i.e., $\tau_{\text{rms}}/T_s$, where $T_s=1/B$ is the sampling period and $\tau_{\text{rms}}=T_s\sqrt{\sum_{d=0}^{D-1} p_d\left(d-\bar{d}\right)^2}$ with $\bar d=\sum_{d=0}^{D-1} p_d d$. The noise power is normalized as $\sigma_{k,n}^2=1,~\forall n\in\mathcal{N},~\forall k\in\mathcal{K}$, so the signal-to-noise ratio (SNR) is defined as $\frac{P_\text{T}}{N}$. 

\subsection{Benchmark Schemes}
To show the effectiveness of the hybrid digital-MiLAC beamforming architecture in wideband MU-MISO systems, we compare it with the following schemes: the conventional fully-digital beamforming \cite{WMMSE}; the conventional hybrid digital-analog beamforming \cite{hybrid_ofdm}, where the analog part is implemented by a fully-connected phase-shifter network; MiLAC-aided beamforming \cite{zheyu_milac_beamforming}, and analog beamforming, where the MiLAC is replaced by a fully-connected phase-shifter network. In the following subsection, the results for fully-digital beamforming are obtained using the WMMSE algorithm in \cite{WMMSE}, the results for hybrid digital-analog beamforming are obtained using the approach in \cite{hybrid_ofdm}, and the results for MiLAC-aided beamforming are obtained using the algorithm in \cite{zheyu_milac_beamforming}.

\subsection{Performance Comparison}
In Fig. \ref{fig:SNR}, we plot the average sum-rate versus SNR for different beamforming schemes. Except for the fully-digital beamforming that has $N_\text{RF}=N_\text{T}$, all other schemes have $N_\text{RF}=K$. Compared with fully-digital beamforming, all other beamforming schemes incur a performance loss, which is in contrast to the narrowband case where the hybrid digital-MiLAC architecture can achieve the same performance as digital beamforming. Among them, hybrid digital-MiLAC performs closest to digital beamforming and consistently outperforms conventional hybrid digital-analog beamforming. For the two fully analog beamforming schemes, MiLAC beamforming outperforms analog beamforming implemented by the fully-connected phase shifters. This is because the MiLAC architecture provides greater flexibility in beamforming design. Overall, the hybrid digital-MiLAC architecture offers a favorable tradeoff between performance and hardware complexity. For example, at SNR $=10$ dB, it achieves $90.52\%$ of the sum-rate to digital beamforming while using only $6.25\%$ of the number of RF chains.

\begin{figure}
    \centering
    \includegraphics[width=0.4\textwidth]{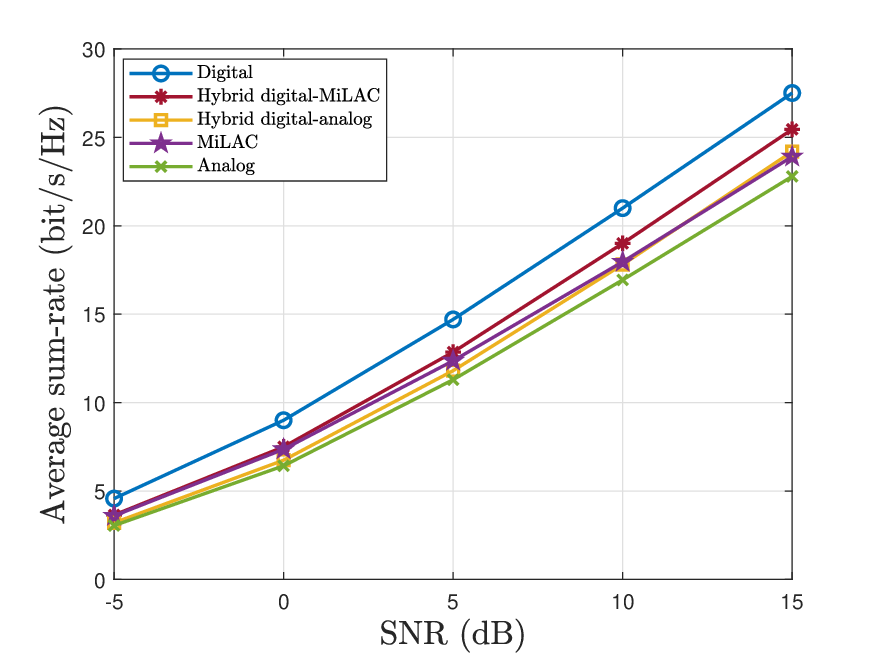}
    \caption{Average sum-rate versus SNR for different beamforming schemes, where $N_\text{RF}=K$ for all schemes except digital beamforming with $N_\text{RF}=N_\text{T}$, and $\epsilon=0.8$.}\label{fig:SNR}
\end{figure}

\begin{figure}
    \centering
    \includegraphics[width=0.4\textwidth]{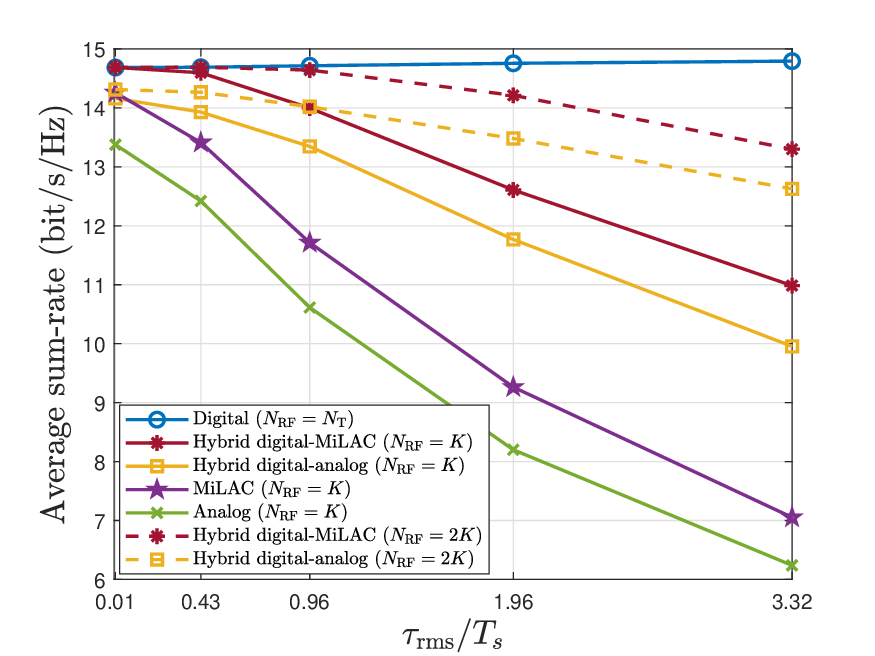}
    \caption{Average sum-rate versus normalized RMS delay spread $\tau_{\text{rms}}/T_s$ for different beamforming schemes ($\text{SNR}=5 \ \text{dB}$).}\label{fig:tau_rms}
\end{figure}

In Fig. \ref{fig:tau_rms}, we investigate the impact of frequency selectivity on the performance of different beamforming schemes by plotting the average sum-rate versus $\tau_{\text{rms}}/T_s$. Different values of $\tau_{\text{rms}}/T_s$ are obtained by varying the PDP decay factor $\epsilon\in\{0.1,\,0.5,\,1,\,2,\,4\}$. It can be observed that when $\epsilon=0.1$, i.e., $\tau_{\text{rms}}/T_s=0.01$, corresponding to a nearly frequency-flat case, MiLAC-aided beamforming and analog beamforming with $N_{\text{RF}}=K$ achieve $97.07\%$ and $91.08\%$ of the fully-digital sum-rate, respectively. The hybrid digital-MiLAC beamforming achieves the same performance as digital beamforming with only $K$ RF chains, which validates the discussions below Remark \ref{remark1}. In contrast, conventional hybrid beamforming still exhibits a performance gap compared with digital beamforming even when  $N_{RF}=2K$. As $\tau_{\text{rms}}/T_s$ increases, i.e., channels become more frequency-selective, the performance of all hybrid and analog beamforming schemes degrades. This is because the analog/MiLAC beamforming matrix is common to all subcarriers, leading to inaccurate beamforming directions that do not match the frequency-dependent channel responses. In particular, purely MiLAC- and analog-aided beamforming perform very poorly in highly frequency-selective channels. However, hybrid digital-MiLAC beamforming can still achieve nearly the same performance as fully-digital beamforming at $\tau_{\text{rms}}/T_s=0.96$. Even for a highly frequency-selective channel with $\tau_{\text{rms}}/T_s=3.32$, it still achieves $89.93\%$ of the fully-digital sum-rate while using only $12.5\%$ of the number of RF chains.

\section{Conclusion}
In this paper, we studied hybrid digital-MiLAC beamforming for MU-MISO-OFDM systems. Under a frequency-flat, lossless, and reciprocal MiLAC model, we characterized the minimum number of RF chains required to realize arbitrary digital beamforming matrices over all subcarriers, and developed a WMMSE-based algorithm for average sum-rate maximization. We showed that the number of RF chains required to achieve digital beamforming flexibility is significantly larger than that in the narrowband case for highly frequency-selective channels. Nevertheless, simulation results showed that hybrid digital-MiLAC beamforming can still achieve comparable performance to digital beamforming with much fewer RF chains, and consistently outperforms hybrid digital-analog beamforming across different levels of frequency selectivity. Overall, this work provides a first step toward MiLAC-aided beamforming in wideband systems and highlights its potential for future large-scale wireless networks. Future work may consider characterizing the frequency-dependent behavior of MiLAC and extending the proposed framework to frequency-selective MiLAC architectures.

\bibliographystyle{IEEEtran}
\bibliography{refs}

\end{document}